\documentclass{ccjnl}
\ccjnlset{
  title = {Joint Beam Scheduling and Power Optimization for Beam Hopping LEO Satellite Systems},
  corr-email = {hszhang@bupt.edu.cn},
}

\author{Shuang Zheng}
\author*{Xing Zhang}
\author{Peng Wang}
\author{Wenbo Wang}

\address{Key Laboratory of Universal Wireless Communication, Ministry of Education, School of Information and Communication Engineering, Beijing University of Posts and Telecommunications, Beijing 100876, China}
\begin{document}
\maketitle

\begin{abstract}
Low earth orbit (LEO) satellite communications can provide ubiquitous and reliable services, making it an essential part of the Internet of Everything network. Beam hopping (BH) is an emerging technology for effectively addressing the issue of low resource utilization caused by the non-uniform spatio-temporal distribution of traffic demands. However, how to allocate multi-dimensional resources in a timely and efficient way for the highly dynamic LEO satellite systems remains a challenge. This paper proposes a joint beam scheduling and power optimization beam hopping (JBSPO-BH) algorithm considering the differences in the geographic distribution of sink nodes. The JBSPO-BH algorithm decouples the original problem into two sub-problems. The beam scheduling problem is modelled as a potential game, and the Nash equilibrium (NE) point is obtained as the beam scheduling strategy. Moreover, the penalty function interior point method is applied to optimize the power allocation. Simulation results show that the JBSPO-BH algorithm has low time complexity and fast convergence and achieves better performance both in throughput and fairness. Compared with greedy-based BH, greedy-based BH with the power optimization, round-robin BH, Max-SINR BH and satellite resource allocation algorithm, the throughput of the proposed algorithm is improved by 44.99\%, 20.79\%, 156.06\%, 15.39\% and 8.17\%, respectively.
\end{abstract}

\section{introduction}
\label{s1}

Satellite communication can provide seamless global communications, meeting the communication needs of sparsely populated or complex terrain areas \cite{1,2}. With the rapid development of broadband Internet services, traditional satellite communications are experiencing challenges such as high traffic demand and low throughput. High-throughput satellites (HTS) increase system capacity with critical technologies such as multi-spot beams, frequency reuse, and antenna gain, which has become a hot spot in today's industry \cite{3,4,5}. Most existing HTS systems use geosynchronous orbit (GEO) satellites, which are simple to build and provide stable and uninterrupted services for fixed areas. In recent years, the research on HTS has progressively evolved into low earth orbit (LEO) satellites. The main reasons are that geosynchronous orbital resources are limited, and LEO satellites have the advantages of low latency and low cost \cite{6}.

In most scenarios, the spatio-temporal distribution of traffic demands changes dynamically. The fixed power allocation and beam scheduling approach will make the limited multi-beam HTS onboard resources unable to respond efficiently to the non-uniform traffic demands. Communication resources with high traffic demand are insufficient, and communication resources with low traffic demand are wasted, leading to inefficient resource utilization. The beam hopping (BH) technology adopts time-slicing technology. That is, $K$ beams are scheduled at each slot  on the basis of traffic demands and $N (N\ge K)$ beam positions can be covered by a LEO satellite simultaneously \cite{7}. In \cite{8}, it is verified that the BH satellite system has better flexibility and can better cope with the non-uniform user distribution and traffic demands, thereby improving the utilization of resources. The proposal of BH technology has attracted extensive attention from researchers. It is considered to be the key technology for developing HTS to very HTS.

In order to better solve the problem of low resource utilization caused by the spatio-temporal distribution of traffic demands, many scholars have proposed a series of BH satellite resource allocation algorithms. Most existing research focuses on the forward link in GEO satellite systems because the forward link is the main direction of service transmission, and resource optimization decision is made totally by the satellite. The studies in \cite{9,10} have solved the problem of resource optimization in BH satellite systems based on global optimization characteristics and applicability of the genetic algorithm (GA). The algorithm in \cite{10} can dynamically adjust the time slot allocation to satisfy the non-uniform traffic demands of each beam position under the influence of time-varying rain attenuation and effectively improve the system's performance. In \cite{11,12}, authors have considered the influence of co-channel interference (CCI) on resource allocation and proposed two iterative algorithms based on maximizing signal to interference plus noise ratio (SINR) and minimizing CCI to improve the total capacity. In \cite{13},  Wang et al. have considered a uniform clustering scheme. And on the basis, to achieve efficient utilization of resources, a joint power and time slots optimization algorithm has been proposed to achieve on-demand resource allocation. In recent years, some scholars have introduced deep reinforcement learning (DRL) methods to optimize resource allocation in BH satellite systems \cite{14,15}. Xu et al. have adopted a novel multi-action selection method based on double-loop learning to optimize resource allocation in the BH satellite systems \cite{14}. The method can ensure the fairness of each beam position, minimize transmission delay of real-time service and maximize throughput of non-real-time service. In \cite{15}, a dynamic beam position illumination and bandwidth allocation scheme based on DRL has been proposed, which flexibly utilizes the three-dimensional resource of time, space and frequency to achieve throughput maximization and fairness considering the time delay between beam positions. 

Due to the high-speed movement of LEO satellites, the dynamic changes in channel conditions and traffic demands make the research on BH resources in GEO satellite systems unable to be directly applied to LEO satellite systems. Consequently, some scholars are exploring resource allocation schemes suitable for BH LEO satellite systems \cite{16,17}. In \cite{16}, Liu et al. have modelled the coverage area as a rectangular block for the time-varying position of LEO satellites and used an iterative algorithm to maximize system capacity. A greedy algorithm has been adopted to improve throughput based on traffic demand in the beam position \cite{17}. Both of the above studies have not discussed the impact of CCI on the LEO satellite systems. Besides, the power of satellite is divided on average, so further optimization is still possible.

Although existing resource allocation techniques for BH LEO satellite systems can bring benefits to system performance, there is still room for improvement. In general, the scenario in discussion is only available for the user distributed in the center of the beam position. In this work, we consider difference in the geographic distribution of sink nodes within the same beam position. In addition, in the BH LEO satellite systems, the CCI cannot be ignored, and the existing work fails to fully consider the influence of the CCI on the system performance. Therefore, we propose an algorithm to improve system performance by using limited onboard resources, which can effectively avoid interference and make timely resource allocation decisions. Last but not least, the power of LEO satellites is limited, but most of the existing work adopts the way that the power is allocated equally among the beams. Thus, the resources are not efficiently utilized. In this work, besides the flexible beam scheduling, the power allocation is optimized. The main work of this paper is as follows:
\begin{itemize}
     \item Firstly, the resource allocation strategy for BH LEO satellite systems is designed. The problem with the joint beam scheduling and power allocation is formulated to minimize the total second-order difference (SOD) cost of traffic offered and traffic demand in each beam position.
     \item Then, the joint beam scheduling and power optimization beam hopping (JBSPO-BH) algorithm is proposed, and the original problem is divided into two sub-problems. One is the problem of beam scheduling, which is solved by the potential game theory. The other is the power optimization problem and the penalty function interior point method is applied to realize power allocation.
     \item Finally, based on simulation results, the JBSPO-BH algorithm can improve the performance of system throughput and fairness. Moreover, it can also converge quickly, which is more likely to adapt to the time-varying characteristics of LEO satellites.
\end{itemize} 

The rest of this paper is organized as follows. Section~\ref{s2} introduces the model and problem formulation of BH LEO satellite systems. Section~\ref{s3} introduces the potential game-based beam scheduling algorithm and JBSPO-BH algorithm. Section~\ref{s4} presents the simulation results. Finally, Section~\ref{s5} concludes the paper.
\section{system model}
\label{s2}
This paper primarily studies the forward link in DVB-S2X \cite{18}, as seen in Figure~\ref{fig1}. Meanwhile, a BH coverage scheme that combines wide and spot beams is adopted \cite{19}. A wide beam, also known as a signalling beam, is responsible for timely signalling transmission and collection of user access and channel state information in the whole coverage area. A spot beam, also known as a service beam, realizes the on-demand high-speed broadband service in the form of BH. The following beam refers to service beam. One LEO satellite can provide $K$ beams to cover the $N$ beam positions in a time-division multiplexing (TDM) manner. In particular, $K$ beams are scheduled, also known as that $K$ beam positions are illuminated simultaneously ($N\ge K$). And the centre of the beam is the centre of the illuminated beam position. There are several sink nodes distributed in each beam position, which collect information from nearby users. We assume that only one sink node can get transmission in a beam during a slot. The packet is transmitted from the gateway to the LEO satellite and then sent to the sink node of the corresponding beam position according to the BH schedule generated by the BH controller. The notations and corresponding descriptions in this section are summarized in Table~\ref{tab1}.

\begin{table}[htbp]
\centering
\caption{Mainly notations.}
\begin{tabular}{llll}\hline

Notation & Description \\\hline
$N$ & The number of beam positions \\
$K$ & The number of beams \\
$P_{n}^t$ &Power allocated to the $n$th beam position \\
&at time $t$ \\
$H_{n,m}^t$ &At time $t$ gain of beam covering the $m$th\\
&beam position to the $n$th beam position \\
$G_T({\theta})$&Gain of the transmit antenna with ${\theta}$\\
&away from the spindle\\
$G_R({\varphi})$&Gain of the receive antenna with ${\varphi}$\\
&away from the spindle\\
$PL_{n}^t$ &Path loss of the $n$th beam position at\\
&time $t$\\
$P_N$& Noise power\\
$SINR_n^t$&SINR of the $n$th beam position at time $t$\\
$B$&Total bandwidth\\
$R_n^t$& Traffic offered of the sink node of the $n$th\\
&beam position at time $t$\\
$\hat{R}_{n}^{t}$ &Traffic demand of the $n$th beam position\\
&at time $t$\\
$\alpha_n^t$&Beam scheduling variable of the $n$th\\
& beam position at time $t$\\
$T_b$ & BH slot\\
$T_p$ & BH cycle\\
$T_s$ & Cycle of updating the location of\\
& subsatellite point\\
$\theta_{3dB}$ &Half-power angle of transmit antenna \\
$\varphi_{3dB}$ &Half-power angle of receive antenna \\\hline
\end{tabular}
\label{tab1}
\end{table}
Figure~\ref{fig2} depicts the time slot scheduling in the BH LEO satellite systems. Let $T_b$ represent the smallest time scale in the system and $T_p$ represent BH cycle. During $T_b$, at most $K$ beams are scheduled while power allocation is achieved.  The BH controller plans the BH schedule and updates the traffic demand of each beam position for packet arrival every $T_p$. The cycle of updating the subsatellite point's location $T_s$ is introduced, since the position of the LEO satellite changes continuously over time.
\begin{figure*}[htbp]
    \centering
    \includegraphics[width=0.8\textwidth]{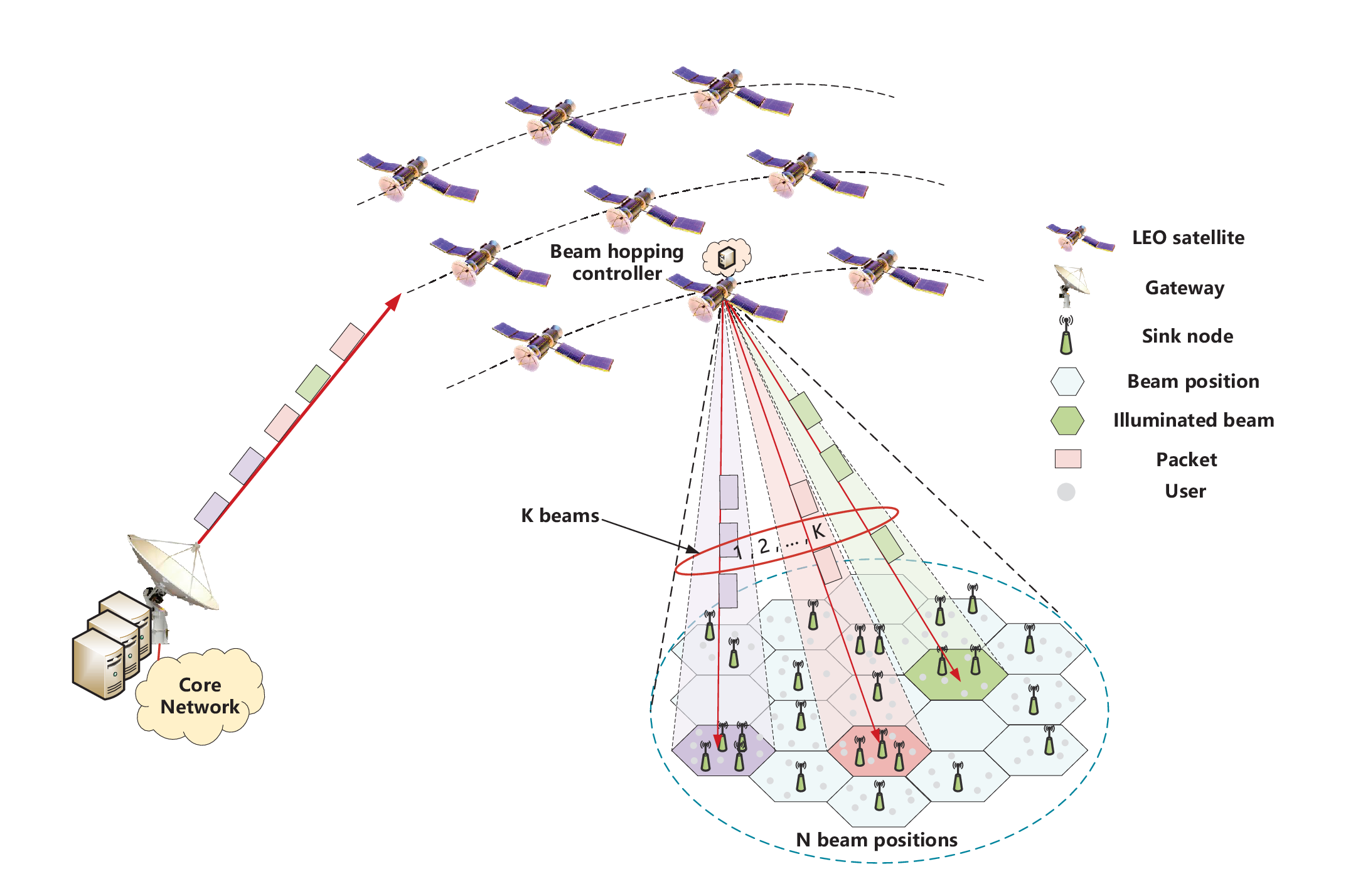}
    \caption{BH LEO satellite scenario.}
	\label{fig1}
\end{figure*}
\begin{figure}[htbp]
    \centering
    \includegraphics[width=0.5\textwidth]{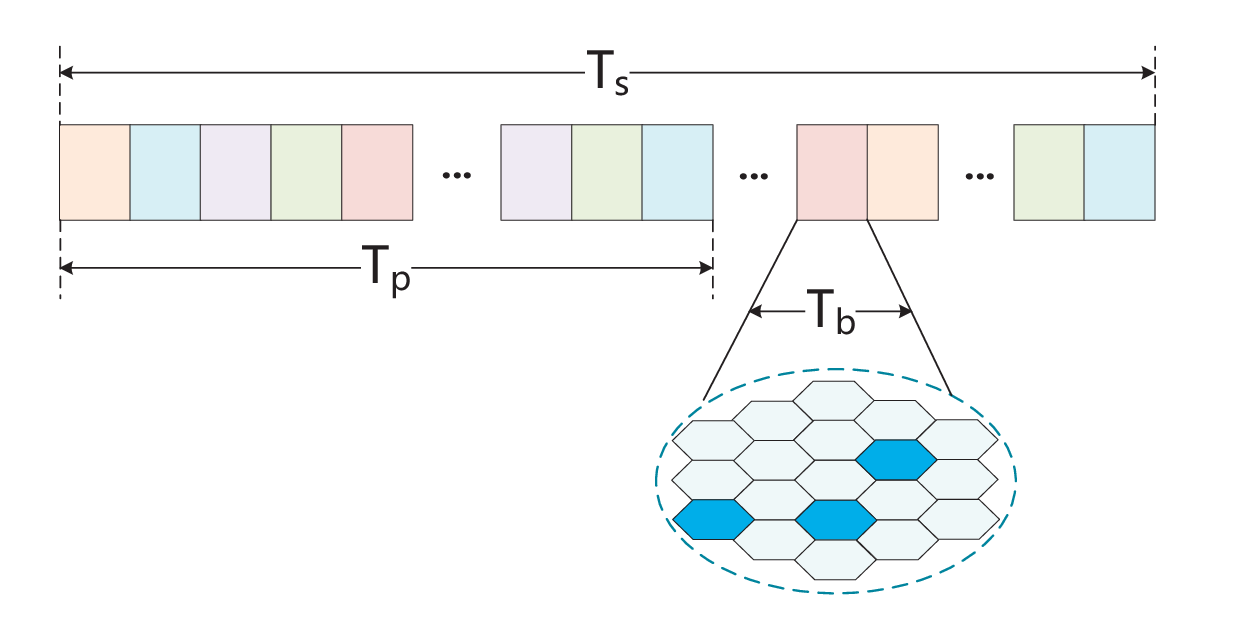}
    \caption{BH LEO satellite time slots scheduling.}
	\label{fig2}
\end{figure}
\subsection{communication model}
The set of sink nodes in the $n$th beam position is ${\mathbf{O}_{n}}=\{{{o}_{n,1}},{{o}_{n,2}},...,{{o}_{{n,{n}_{\max }}}}\}$, and $n_{max}$ is the number of sink nodes in the $n$th beam position. In this paper, it is assumed that only one sink node in a beam position can get transmission service during $T_b$. At time $t$, the signal received by the served sink node of the $n$th beam position is:
\begin{equation}
y_{n}^{t}=\sqrt{P_{n}^{t}H_{n,n}^{t}}{{x}_{n}^t}+\sum\limits_{m\ne n}{\sqrt{P_{m}^{t}H_{n,m}^{t}}{{x}_{m}^t}+{{n}_{0}}},
\label{1}
\end{equation}
where $P_{n}^t$ is the transmit power of the LEO satellite allocated to the $n$th beam position, and $x_{n}^t$ is the transmit symbol of the $n$th beam position at time $t$. In addition, $n_0$ is the additional white Gaussian noise with power of $P_N$. $H_{n,m}^t$ denotes channel gain of the beam covering the $m$th beam position to the $n$th beam position, specifically:
\begin{equation}
H_{n,m}^{t}={{G}_{T}}(\theta _{n,m}^{t}){{G}_{R}}(\varphi _{n}^{t})PL_{n}^{t}.
\label{2}
\end{equation}

In Eq.~\eqref{2}, $\theta_{n,m}^t$ is the angular position of the selected sink node of $n$th beam position from the beam covering the $m$th beam position at time $t$. Similarly, $\varphi_{n}^t$ is the angle between the signal transmission's direction and the sink node's spindle of the $n$th beam position at time $t$. Thus, $G_T({\theta_{n,m}^t})$ and $G_R({\varphi_{n}^t})$ represent the gain of the transmit antenna with $\theta_{n,m}^t$ away from the spindle and the receive antenna with $\varphi_{n}^t$ away from the spindle, respectively. $PL_{n}^t$ is the path loss of the $n$th beam position at time $t$, including free space loss and other loss ${PL}_o$. The following formula calculates the gains of transmit and receive antennas \cite{20,21}:
\begin{equation}
G(\theta )={{G}_{\max }}{\Big(\frac{{{J}_{1}}(u(\theta ))}{2u(\theta )}+36\frac{{{J}_{3}}(u(\theta ))}{{{(u(\theta ))}^{3}}}\Big)^{2}},
\label{3}
\end{equation}
\begin{equation}
u(\theta )=\frac{2.07123\sin \theta }{\sin {{\theta }_{3dB}}},
\label{4}
\end{equation}
where $G_{max}$ is the maximum gain of the antenna, and $J_1$ and $J_3$ represent the first-order and third-order Bessel functions, respectively. Besides, $\theta$ is the off-axis angle, and ${\theta}_{3dB}$ is the half-power angle of the antenna. Therefore, the SINR of the sink node of the $n$th beam position at time $t$ can be expressed as:
\begin{equation}
SINR_{n}^{t}=\frac{\alpha _{n}^{t}P_{n}^{t}H_{n,n}^{t}}{\sum\limits_{m\ne n}{\alpha _{m}^{t}P_{m}^{t}H_{n,m}^{t}+{{P}_{N}}}},
\label{5}
\end{equation}
where $\alpha_n^t$ is a beam scheduling variable indicating whether the $n$th beam position is illuminated at time $t$.

Therefore, the corresponding traffic offered of the $n$th beam position during $T_b$  is given by the following Shannon's formula:
\begin{equation}
R_{n}^{t}=B{{T}_{b}}{{\log }_{2}}(1+SINR_{n}^{t}),
\label{6}
\end{equation}
where $B$ is total bandwidth.
\subsection{problem formulation}
Considering that in BH LEO satellite systems, the resource allocation should not only improve the system throughput, but also ensure the fairness among beam positions. Therefore, the metric of SOD cost of traffic offered and traffic demand is adopted. To obtain the minimum SOD cost of the whole system, the optimization problem can be formalized as \cite{22}:
\begin{equation}
\min \sum\limits_{n}{{{(\sum\limits_{t}{R_{n}^{t}-{{{\hat{R}}}_{n}}})}^{2}}},
\label{7}
\end{equation}
where ${\hat{R}}_{n}$ is traffic demand of the $n$th beam position. Further, the above problem can be transformed into optimizing SOD cost in each BH slot.
In problem~\eqref{8}, $\hat{R}_{o_{n,i}}^{t}$ and $\hat{R}_{n}^{t}$ represent the traffic demand of sink node $o_{n,i}$ and the $n$th beam position at time $t$, respectively. Let $o_{n,i^{*}}$ represent the sink node for packet transmission in the $n$th beam position at time $t$. $C1$ means that the number of illuminated beam positions in the same slot cannot exceed the number of beams provided by the LEO satellite; $C2$ indicates that the upper limit of the sum of power allocated to each beam position is the maximum transmit power of the satellite; $C3$ indicates that the beam scheduling variable $\alpha _{n}^{t}$ is a binary integer, which means whether a beam is scheduled to the $n$th beam position at time $t$; $C4$ is the relationship between $\hat{R}_{o_{n,i}}^{t}$ and $\hat{R}_{n}^{t}$; $C5$ ensures that resources are not wasted. 
\begin{equation}
\begin{aligned}
&\underset{\alpha _{n}^{t},P_{n}^{t}}{\mathop{\min }} \sum\limits_{n=1}^{N}{(R_{n}^{t}-}\hat{R}_{n}^{t}{{)}^{2}}, \\ 
&\begin{aligned}
s.t.\  &C1  :  \sum\limits_{n=1}^{N}{\alpha _{n}^{t}}\le K, \\ 
&C2  :  \alpha _{n}^{t}\in \{0,1\},\\
&C3  :  \sum\limits_{n=1}^{N}{P_{n}^{t}}\le {{P}_{\max }}, \\ 
&C4  :  \hat{R}_{n}^{t}=\sum\limits_{{{o}_{n,i}}\in {{\mathbf{S}}_{n}}}{\hat{R}_{{{o}_{n,i}}}^{t}},\\
&C5  :  {R}_{n}^{t}\le\hat{R}_{{{o}_{n,i^*}}}^{t}. 
\end{aligned}
\end{aligned}
\label{8}
\end{equation}

Since the above problem is a mixed integer programming problem, it is difficult to be solved by the convex optimization method, so the problem is divided into two sub-problems.

\section{joint beam scheduling and power optimization beam hopping}
\label{s3}

The above optimization problem is a mixed integer programming problem which is challenging to solve because the beam scheduling variables are binary integer variables and the power allocation variables are continuous variables. As a consequence, the proposed algorithm decouples the optimization problem into two sub-problems. For the beam scheduling problem, we prove that it is a potential game problem and use game theory to solve it. Based on the beam scheduling scheme, the penalty function interior point method is further used to optimize the power allocation.
\subsection{beam scheduling}
Considering that the transmit power of the satellite is allocated to each beam averagely, where ${{P}_{ave}}=\frac{{{P}_{\max }}}{K}$. problem(\ref{8}) can be rewritten as
\begin{equation}
\begin{aligned}
&\underset{\alpha _{n}^{t}}{\mathop{\min }} \sum\limits_{n=1}^{N}{{{\Big(B{{T}_{b}}{{\log }_{2}}(1+SINR_{n}^{t})-\hat{R}_{n}^{t}\Big)}^{2}}},\\ 
&\begin{aligned}
s.t.\ &C1 : \sum\limits_{n=1}^{N}{\alpha _{n}^{t}}\le K,\\
&C2 : \alpha _{n}^{t}\in \{0,1\},\\
&C3 :SINR_{n}^{t}=\frac{\alpha _{n}^{t}{{P}_{ave}}H_{n,n}^{t}}{\sum\limits_{m\ne n}{\alpha _{m}^{t}{{P}_{ave}}H_{n,m}^{t}+{{P}_{N}}}}.
\end{aligned}
\end{aligned}
\label{9}
\end{equation}

Referring to \cite{23,24}, the game is denoted by $\mathbf{G}=\{\mathbf{K},{{({{a}_{k}})}_{\{k\in \mathbf{K}\}}},{{({{u}_{k}}({{a}_{k}},{{a}_{-k}}))}_{\{k\in\mathbf{K}\}}}\}$. Let $\mathbf{K}=\{1,2,...,K\}$ represent the set of all game players. Define ${{a}_{k}}\in \mathbf{B}\subseteq \mathbf{N}=\{0,1,...,N\}$ as the beam scheduling strategy of the $k$th beam, where $\mathbf{B}$ is the set of beam positions with non-zero traffic demand. Besides, ${{a}_{-k}}$ is defined as the beam scheduling strategy of other beams except the $k$th beam, namely, ${{a}_{-k}}={{({{a}_{{{k}^{'}}}})}_{{{k}^{'}}\in \mathbf{K}}}\backslash \{{{a}_{k}}\}$. In the game, each player expects to minimize its utility function ${{{{u}_{k}}({{a}_{k}},{{a}_{-k}})}}$ \cite{25}.
\begin{equation}
\begin{aligned}
  & \underset{{{a}_{k}}}{\mathop{\min }}\,{{u}_{k}}({{a}_{k}},{{a}_{-k}}),\\ 
&\begin{aligned}
s.t. \  &C1  :  {{a}_{k}}\in \mathbf{B}, \\ 
&C2  :  {{a}_{k}}\notin  {{a}_{-k}},     \\ 
\end{aligned}
\end{aligned}
\label{10}
\end{equation}
where ${{{{u}_{k}}({{a}_{k}},{{a}_{-k}})}}$ is defined in Eq.~\eqref{12}.
\begin{strip}
\begin{equation}
\begin{aligned}
  & {{u}_{k}}({{a}_{k}},{{a}_{-k}})=\sum\limits_{k}{\Bigg({{\bigg(B{{T}_{b}}{{\log }_{2}}\Big(1+\frac{{{P}_{ave}}H_{{{a}_{k}},{{a}_{k}}}^{t}}{\sum\limits_{l\ne k}{{{P}_{ave}}H_{{{a}_{k}},{{a}_{l}}}^{t}+{{P}_{N}}}}\Big)\bigg)}^{2}}-2B{{T}_{b}}\hat{R}_{{{a}_{k}}}^{t}{{\log }_{2}}\Big(1+\frac{{{P}_{ave}}H_{{{a}_{k}},{{a}_{k}}}^{t}}{\sum\limits_{l\ne k}{{{P}_{ave}}H_{{{a}_{k}},{{a}_{l}}}^{t}+{{P}_{N}}}}\Big)\Bigg)}. \\ 
\end{aligned}
\label{12}
\end{equation}
\end{strip}
\begin{theorem}\label{Theorem1}Game $\mathbf{G}$ is an exact potential game under potential function F.
\label{theo1}
\end{theorem}
\begin{proof}Please see Appendix A and B.
\end{proof}

\begin{algorithm}[htbp]
\caption{Potential game-based beam scheduling algorithm.}
\begin{algorithmic}[1]
\REQUIRE{${{\mathbf{B}}^{t}},{{\mathbf{H}}^{t}},{{\{\hat{R}_{n}^{t}\}}_{n\in {\mathbf{B}^{t}}}}$}
\ENSURE{${{\mathbf{a}}^{t}}^{*}$}
\STATE Initialize $iteration=1$.
\STATE A random beam scheduling scheme $\mathbf{a}$ is adopted after assigning an average power allocation.
\WHILE{$iteration<{iteration}_{max}$} 
    \FOR {beam $k=1,2,...,K$}
    	\STATE Calculate utility function ${{u}_{k}}({{a}_{k}},{{a}_{-k}})$.
     \STATE Illuminate beam position based on best response principle
      \begin{equation}
      a_{k}^{t}=\underset{{{a}_{k}}\in {\mathbf{B}}, {{a}_{k}}\notin  {{a}_{-k}}}{\mathop{\arg \min }}\,{{u}_{k}}({{a}_{k}},{{a}_{-k}}).
     \nonumber 
	\end{equation}
    \ENDFOR
    \STATE $iteration\leftarrow iteration+1$.
    \STATE Until ${({\mathbf{a}})}^{iteration}={({\mathbf{a}})}^{iteration-1}$.
\ENDWHILE
\STATE${{\mathbf{a}}^{t}}^{*}\leftarrow{({\mathbf{a}})}^{iteration}$.
\end{algorithmic}
\label{algorithm1}
\end{algorithm}
According to Theorem \ref{theo1}, the game is an exact potential game. According to \cite{26}, there must be a pure strategy Nash equilibrium (NE) point, which is the global or local optimal solution of the potential function. 

In summary, the potential game-based beam scheduling algorithm is described in Algorithm \ref{algorithm1}. At time $t$, the beam position set $\mathbf{B}^t$ with non-zero traffic demand, channel gain set $\mathbf{H}^t$, and traffic demand set ${{\{\hat{R}_{n}^{t}\}}_{n\in {\mathbf{B}^{t}}}}$ are generated first. The initial beam scheduling set $\mathbf{a}$ is obtained randomly while the power is allocated averagely. Then, in each iteration, each player optimizes its utility function based on the best response principle, until NE point is reached.
\subsection{power optimization}
Furthermore, the power allocation is carried out on the basis of the beam scheduling scheme. In this case the original problem~\eqref{8} can be rewritten as problem~\eqref{14}. It is obvious that the above problem only contains continuous variables. Due to the coupling of variables in the optimization function, the problem~\eqref{14} remains a non-convex optimization problem despite being easier to handle than the original problem~\eqref{8}. Although there exist nonlinear constraints, the penalty function interior point method can be used to solve it.
\begin{equation}
\begin{aligned}
&\underset{P_{{{a}_{k}}}^{t}}{\mathop{\min }}\sum\limits_{k=1}^{K}{{{(R_{{{a}_{k}}}^{t}-\hat{R}_{{{a}_{k}}}^{t})}^{2}}} ,\\ 
&\begin{aligned}
s.t.\ &C1: \sum\limits_{k}{P_{{{a}_{k}}}^{t}\le {{P}_{\max }}},\\ 
&C2:P_{{{a}_{k}}}^{t}\ge 0, \\ 
&C3:R_{a_k}^t\le\hat{R}_{{{a}_{k}}}^{t},\\
&C4:R_{{{a}_{k}}}^{t}=BT_b{\log}_2\Big(1+\frac{P_{{{a}_{k}}}^{t}H_{{{a}_{k}},{{a}_{k}}}^{t}}{\sum\limits_{l\ne k}{P_{{{a}_{l}}}^{t}H_{{{a}_{k}},{{a}_{l}}}^{t}+{{P}_{N}}}}\Big).
\end{aligned}
\end{aligned}
\label{14}
\end{equation}

In particular, there is an approximation problem for $\mu$ in each iteration \cite{27,28}.
\begin{equation}
\begin{aligned}
 &\underset{{\mathbf{P}^{t}},\mathbf{s}}{\mathop{\min }}\,{{f}_{\mu }}=g({\mathbf{P}^{t}})-\mu (\sum\limits_{i}{\ln ({{s}_{i}})}), \\ 
&\begin{aligned}
s.t.\ &C1:\sum\limits_{k}{P_{{{a}_{k}}}^{t}}-{{P}_{\max }}+{{s}_{1}}=0, \\ 
&C2:-P_{{{a}_{k}}}^{t}+{{s}_{k+1}}=0, \\ 
&C3:R_{a_k}^t-\hat{R}_{{{a}_{k}}}^{t}+s_{k+1+K}=0.
\end{aligned}
\end{aligned}
\label{15}
\end{equation}

In Eq.~\eqref{15}, ${{\mathbf{P}}^{t}}$ is the scheme of power optimization and $g({\mathbf{P}}^{t})$ is the objective function in problem~\eqref{14} after substituting the equality constraint $C4$. In addition, the size of the slack variable set $\mathbf{s}$ is equal to the number of inequality constraints and each variable is restricted to be positive to keep the value within the feasible domain. As the number of iterations increases, $\mu$ gradually decreases to 0, and the minimum value of $f_{\mu}$ should be close to the minimum value of $g({\mathbf{P}}^{t})$ correspondingly. The increasing logarithmic term is the penalty function. In this case, the approximation problem consists of a series of equality constraint problems, which is simpler to solve than the original problem.

The solution of power optimization problem is as follows. The average power allocation strategy is taken as the initial value point. Then, the Newton step is used to solve the approximation problem by linear approximation. If the Newton step fails, the trust region method is applied and the conjugate gradient step is used. In each iteration, the problem~\eqref{15} is optimized and the value of $\mu$ is updated. Until the accuracy requirement is met, the iteration stops. Finally, the global or local optimal value is obtained, that is, the power optimization scheme ${\mathbf{P}^{t^*}}$ is obtained \cite{29}.
\subsection{Joint Beam Scheduling and Power Optimization}
\begin{algorithm}[htbp]
\caption{Joint beam scheduling and power optimization beam hopping algorithm.}
\begin{algorithmic}[1]
\REQUIRE{$N,K,{{P}_{\max }},{{t}_{\max }},$ latitude and longitude coordinates of every beam position's centre and sink nodes}
\ENSURE{${{\{\hat{R}_{n}^{{{t}_{\max }}}\}}_{n\in \{1,2,...,N\}}}$}
\STATE Initialize the location of subsatellite point.
\STATE Initialize ${{t}_{1}}=0$ and ${{\{\hat{R}_{n}^{{{t}_{1}}}\}}_{n\in \{1,2,...,N\}}}$.
\WHILE {$t_1<t_{max}$}
	\STATE Update $\hat{R}_{o_{n,i}}^{t_1}$ with packet arrival.
	\IF {$mod(t_1,T_s)==0$}
		\STATE Update the location of subsatellite point.
     \ENDIF
     \STATE $t_2=0$.
	\WHILE{$t_2<T_p$} 
		\STATE Record preselected set $\mathbf{B}^{t_1+t_2}$.
		\STATE Select sink nodes based on greedy strategy\\
	     \begin{equation*}
          o_n^*=\underset{{{o}_{n,i}}\in \mathbf{O}_n}{\mathop{\arg \max }}\,\hat{R}_{o_{n,i}}^{t_1+t_2}.
		\end{equation*}
		\STATE Calculate ${\mathbf{H}}^{t_1+t_2}$.
		\IF{ $length(\mathbf{B}^{t_1+t_2})<K$}
			\STATE Illuminate each beam position in $\mathbf{B}^{t_1+t_2}$.
		\ELSE
      		\STATE Adopt Algorithm \ref{algorithm1} and get ${\mathbf{a}^{({{t}_{1}}+{{t}_{2}})}}^{*}$.
			\STATE Optimize power allocation with penalty function interior point method and get ${\mathbf{P}^{({{t}_{1}}+{{t}_{2}})}}^{*}$.
		\ENDIF
		\STATE Calculate ${{\{R_{n}^{{{t}_{1}}+{{t}_{2}}}\}}_{n\in \{{\mathbf{a}^{({{t}_{1}}+{{t}_{2}})}}^{*}\}}}$.
		\STATE $t_2 \leftarrow t_2+T_b$.
		\STATE Update $\hat{R}_{n}^{{{t}_{1}}+{{t}_{2}}}$.
	    
	\ENDWHILE
	\STATE $t_1\leftarrow t_1+T_p$.
\ENDWHILE
\STATE Get ${{\{\hat{R}_{n}^{{{t}_{\max }}}\}}_{n\in \{1,2,...,N\}}}$.
\end{algorithmic}
\label{algorithm2}
\end{algorithm}
In the scenario described in this paper, multiple sink nodes are distributed in each beam position, and their channel conditions vary. The greedy strategy is used to select sink nodes in order to make more effective use of the limited resources of the satellite. That is, the sink node with the highest traffic demand in the present slot is selected for the corresponding packet transmission.

The JBSPO-BH algorithm is shown in Algorithm \ref{algorithm2}. Considering the high-speed motion characteristics of the LEO satellite, the location of the subsatellite point are updated every $T_s$. Firstly, the traffic demand of each beam position is obtained, and the beam positions with non-zero traffic demand consists of a preselected set. The policy search space can be effectively reduced by recording the preselected set. The sink node for data transmission in the present slot is selected in each beam position, and then the channel gain matrix is calculated. If the number of beam positions in the preselected set is less than the number of beams a single satellite can provide, then all of the beam positions will be illuminated. Otherwise, firstly, the beam scheduling scheme is obtained based on Algorithm \ref{algorithm1}. Then the corresponding power optimization is solved based on the penalty function interior point method. Finally, the traffic demand of each beam position is updated for the next slot. Moreover, BH controller makes a decision and the service arrival is counted every BH cycle in the BH LEO satellite systems.

\section{simulation}
\label{s4}
\begin{table}[htbp]
\centering
\caption{Simulation parameters \cite{30,31}.}
\begin{tabular}{cccc}\hline
Parameter & Value \\\hline
$B$ & 200 MHz \\
$T_b$ & 0.5 ms \\
$T_p$ &20 ms\\
$T_s$ & 200 ms\\
$T$ &20 s\\
${G}_{T\max}$ &36.2 dBi\\
${G}_{R\max}$ &20 dBi\\
${PL}_O$ & 7 dB\\
$P_N$ & $7.96\times {{10}^{-13}}$ W \\
$P_{max}$ & $250$ W\\
$M$& 10 kbits \\\hline
\end{tabular}
\label{tab2}
\end{table}
In this section, numerical results show the effectiveness of JBSPO-BH algorithm in the BH LEO satellite systems. The orbital altitude of the LEO satellite is 508 km, and it can provide 8 beams to cover 61 beam positions in a TDM way. The sink node is based on the Poisson point process distribution with the density ${\lambda }_{s}=1.5\times {{10}^{-8}}$. For simplicity, packets are assumed to be in a fixed size $M$ and arrive in a Poisson model with $\lambda$. Then, we adopt the Ka-band with the downlink frequency of 20 GHz. In addition, the basic parameters of GA-BH are: the number of generation $G$ is 200, the number of population $P$ is 100, the mutation probability $P_{mut}$ is 0.2, and the crossover probability $P_{cro}$ is 0.8. More specific simulation parameters are shown in Table \ref{tab2} \cite{30,31}.

In this paper, we compare the proposed algorithm with the following 5 different BH resource allocation algorithms and satellite resource allocation algorithm in \cite{20}.

\emph{1) The Greedy-Based Beam Hopping (G-BH):} The G-BH algorithm selects the top $K$ beam positions to illuminate based on the highest traffic demand in each BH slot. In addition, the power is allocated to the illuminated beam positions on average.

\emph{2) The Greedy-Based Beam Hopping with Power Optimization (G-BHPO):} The beam scheduling scheme is the same as the G-BH algorithm, and the power allocation strategy is optimized using the penalty function interior point method described in this paper.

\emph{3) The Round-Robin Beam Hopping (RR-BH):} The RR-BH algorithm is that all beam positions occupy $K$ beam resources in turn, regardless of the traffic demand of different beam positions. And the transmit power of each beam is equal.

\emph{4) The Max-SINR Beam Hopping (Max-SINR-BH):} The Max-SINR algorithm is focused on maximizing the SINR of each BH slot to satisfy the traffic demand, and the power is equally allocated to the beam.

\emph{5) The Genetic Algorithm Beam Hopping (GA-BH):} The beam scheduling scheme is solved by GA, and the power optimization strategy based on penalty function interior point method described in this paper is adopted.
\subsection{Convergence}
\begin{figure}[htbp]
    \centering
    \includegraphics[width=0.5\textwidth]{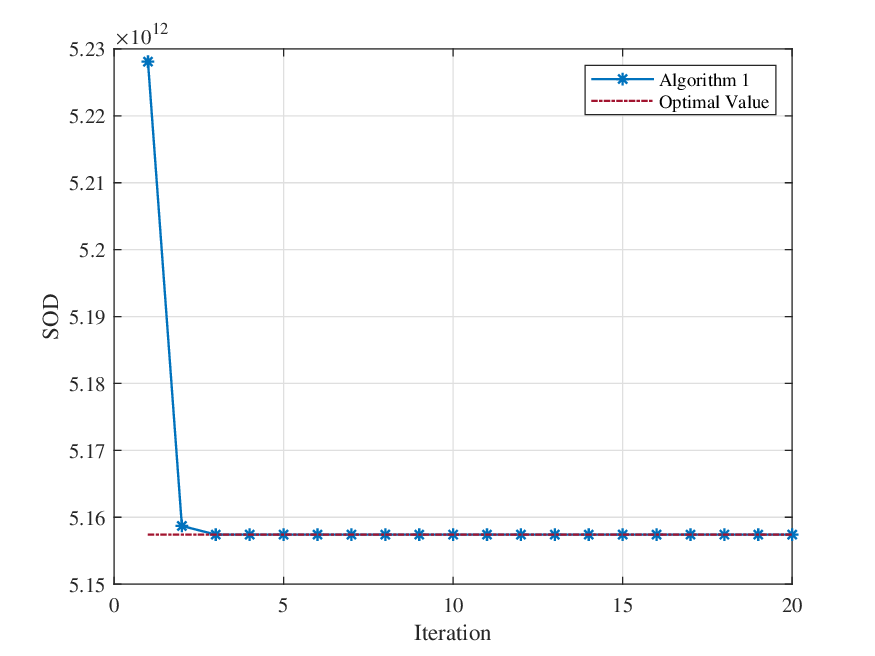}
    \caption{SOD decreasing by iterations.}
	\label{fig3}
\end{figure}

Figure~\ref{fig3} depicts that as the number of iterations increases, the SOD cost gradually decreases. It means the system performance is improved with the potential game-based beam scheduling algorithm shown in Algorithm \ref{algorithm1}. After several iterations, Algorithm \ref{algorithm1} converges to NE point, so the beam scheduling state reaches the equilibrium state. As shown in Figure~\ref{fig3}, the algorithm can converge quickly and is better suited to the dynamics of BH LEO satellite systems.
\subsection{system performance}
\begin{figure}[htbp]
    \centering
    \includegraphics[width=0.5\textwidth]{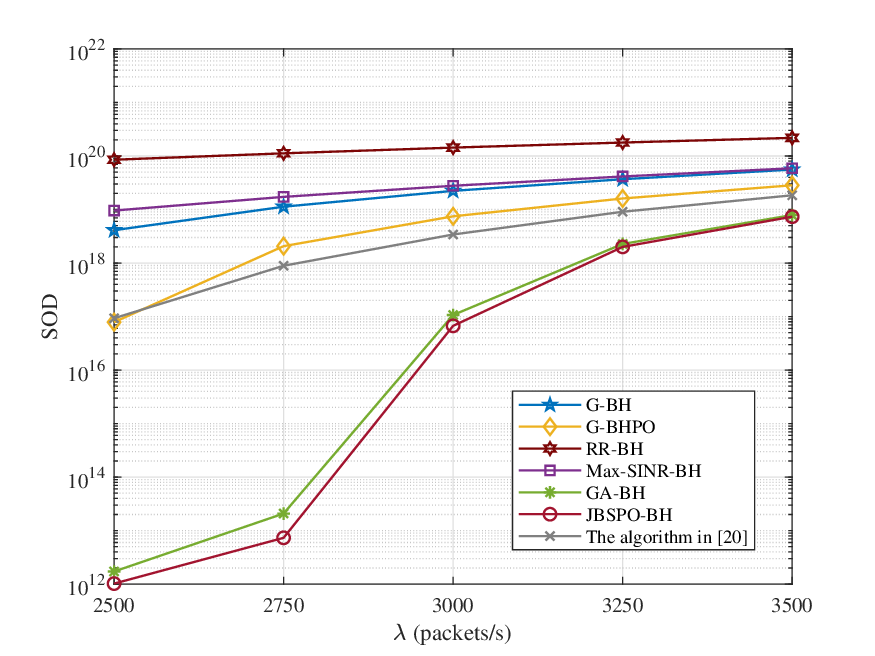}
    \caption{SOD cost versus different value of $\lambda$ in the BH LEO satellite system.}
	\label{fig4}
\end{figure}
In this subsection, we use SOD cost, average throughput, Jain fairness index (JFI) and Satisfaction to evaluate and compare the performance of the above algorithms.

Figure~\ref{fig4} shows the SOD cost varies with different traffic demands.  The SOD cost gradually increases with the increase of $\lambda$, because it's more difficult for LEO satellite to provide sufficient resource when the traffic demands become higher. It can be seen obviously that the SOD cost of the JBSPO-BH algorithm is significantly lower than that of other 6 different comparison algorithms. Therefore, the proposed algorithm can better match the non-uniform traffic demands among different beam positions with limited resources. The SOD cost of G-BHPO algorithm is also significantly reduced compared with G-BH algorithm, proving the effectiveness of the power optimization method.

\begin{figure}[htbp]
    \centering
    \includegraphics[width=0.5\textwidth]{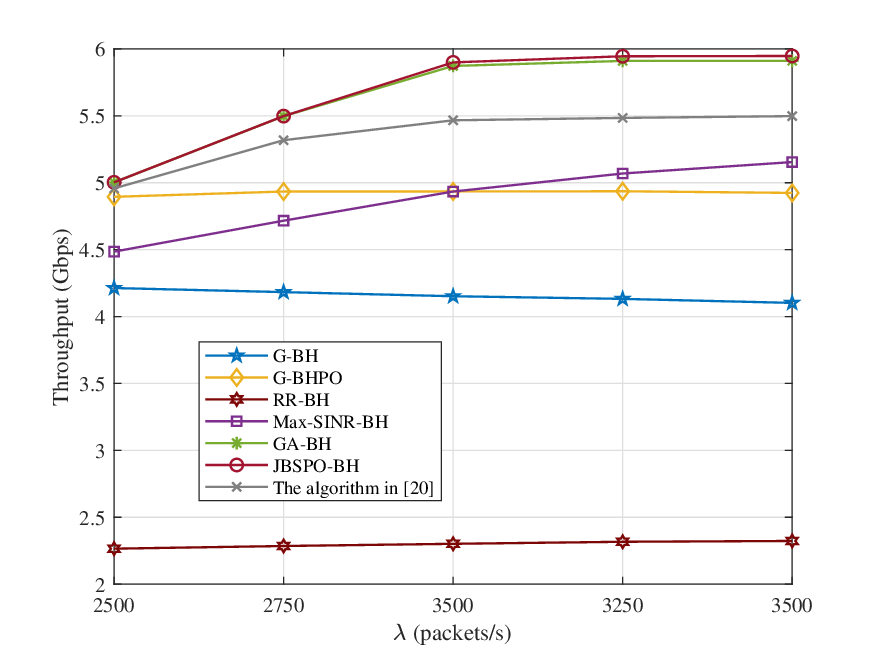}
    \caption{Throughput versus different value of $\lambda$ in the BH LEO satellite system.}
	\label{fig5}
\end{figure}
\begin{figure}[htbp]
    \centering
    \includegraphics[width=0.5\textwidth]{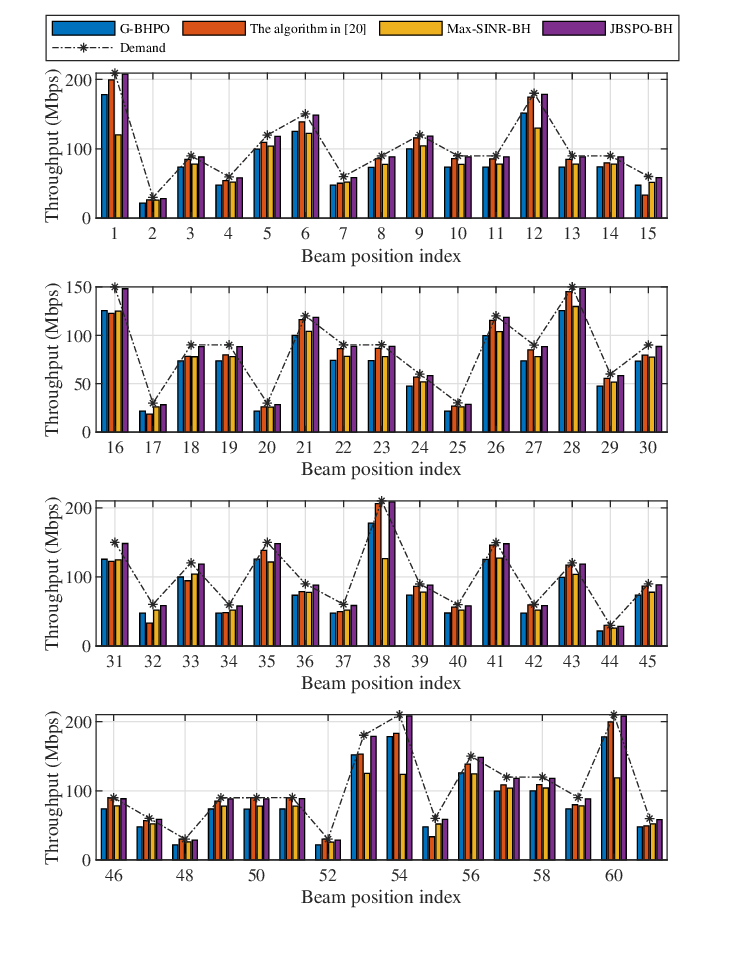}
    \caption{The throughput of each beam position in the BH LEO satellite system.}
	\label{fig6}
\end{figure}
\begin{figure}[htbp]
    \centering
    \includegraphics[width=0.5\textwidth]{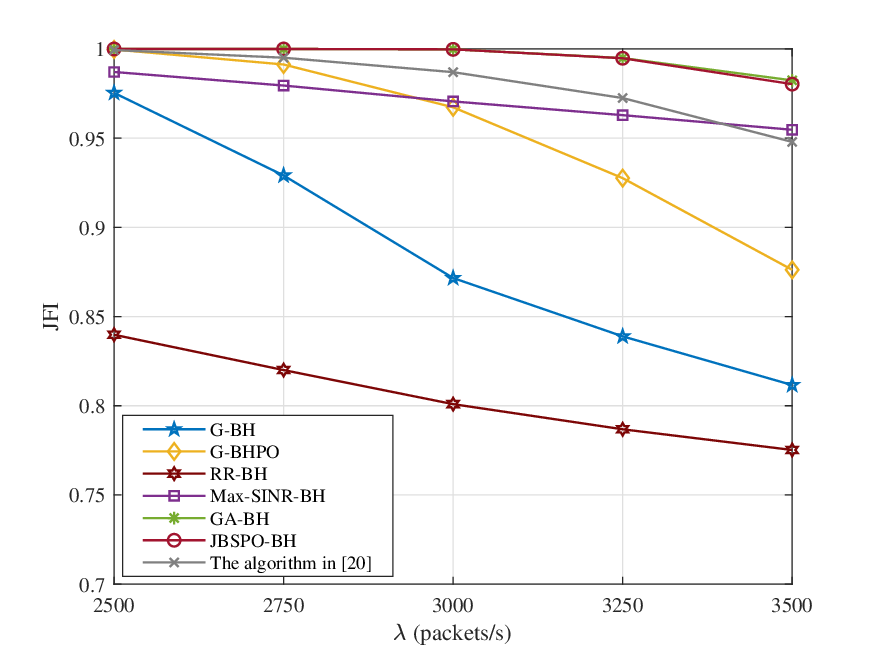}
    \caption{JFI versus different value of $\lambda$ in the BH LEO satellite system.}
	\label{fig7}
\end{figure}
Figure~\ref{fig5} and Figure~\ref{fig6} describe system's throughput using different algorithms. As shown in Figure~\ref{fig5}, as $\lambda$ increases, the JBSPO-BH algorithm can achieve the system throughput of about 5.9 Gbps. The performance of the proposed algorithm is similar to that of GA-BH algorithm. Compared with G-BH algorithm, G-BHPO algorithm, RR-BH algorithm, Max-SINR-BH algorithm and the algorithm in \cite{20}, the proposed algorithm can improve the system throughput by 44.99\%, 20.79\%, 156.06\%, 15.39\% and 8.17\%. This is because other algorithms except GA-BH algorithm and the proposed algorithm do not fully consider the influence of CCI. Since G-BHPO algorithm is obviously superior to G-BH algorithm, RR-BH has the worst performance and GA-BH algorithm has similar performance to the proposed algorithm, Figure~\ref{fig6} only shows the performance comparison of the 4 algorithms with $\lambda =3000$. It can be seen that the JBSPO-BH algorithm can better achieve performance improvement of system throughput.

Figure~\ref{fig7} and Figure~\ref{fig8} describe the fairness performance of different algorithms analyzed by using the satisfaction and JFI. The the satisfaction of the $n$th beam position $S_n$ is defined as:
\begin{equation}
S_n=\frac{R_n}{\hat{R}_n}.
\label{16}
\end{equation}

And the JFI is defined as \cite{32}:
\begin{equation}
JFI=\frac{{{\Big(\sum\limits_{n=1}^{N}{{{S}_{n}}}\Big)}^{2}}}{N\sum\limits_{n=1}^{N}{{\big(S_n\big)^{2}}}}.
\label{17}
\end{equation}

\begin{figure}[htbp]
    \centering
    \includegraphics[width=0.5\textwidth]{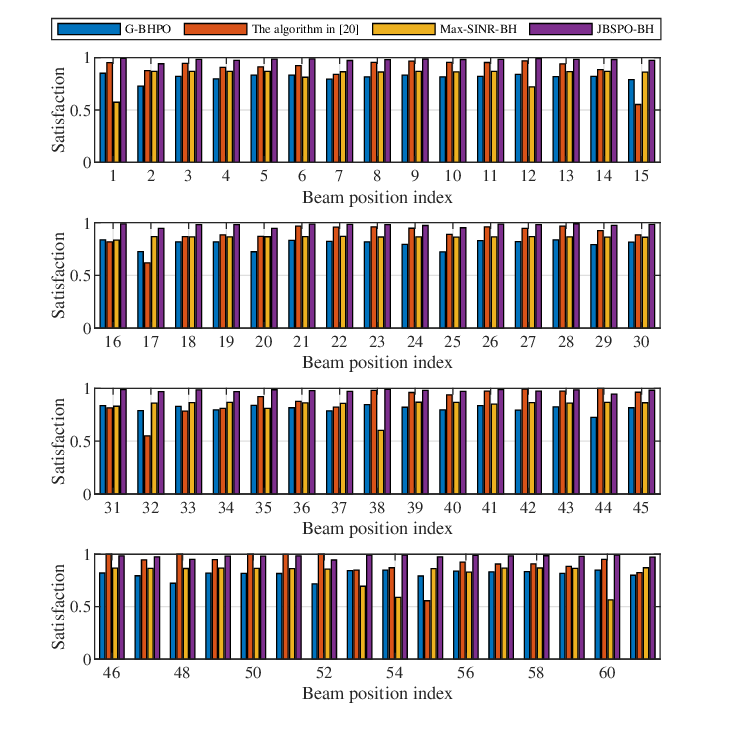}
    \caption{The satisfaction of each beam position in the BH LEO satellite system.}
	\label{fig8}
\end{figure}

In Eq.~\eqref{16}, $R_n$ represents throughput of the $n$th beam position. Figure~\ref{fig7} shows the JFI of the proposed algorithm and other algorithms versus different value of $\lambda$. It is obvious that the proposed algorithm achieve higher fairness than other algorithms. In Figure~\ref{fig8}, the satisfaction of each beam position has been depicted with $\lambda =3000$. In some beam positions, the algorithm in \cite{20} achieves better satisfaction, but for the whole system, the JBSPO-BH algorithm achieves a more balanced improvement in satisfaction, achieving better fairness performance.
\subsection{complexity analysis}
Based on the above results, the performance of JBSPO-BH algorithm is comparable to that of the GA-BH algorithm and significantly superior to that of other algorithms. Additionally, this subsection only focuses on the time complexity analysis of two algorithms' beam scheduling scheme, because both of the power allocation is optimized using penalty function interior point method described in this paper. It is analysed that the time complexity of GA-BH algorithm can be expressed as:
\begin{equation}
\begin{aligned}
  & {{C}_{GA-BH}} \\ 
 & =G*(P*{{P}_{mut}}+P*{{P}_{cro}})*O(fitness) \\ 
 & \approx G*P*O(fitness), \\ 
\end{aligned}
\label{18}
\end{equation}
where $O(fitness)$ is the time complexity of fitness function in GA-BH algorithm. And $O(fitness)=O({{K}^{2}})$, determined by the Eq.~\eqref{8}. It's necessary to note that the performance of GA-BH algorithm is directly influenced by $G$ and $P$. 

The time complexity of the JBSPO-BH algorithm proposed in this paper can be expressed as:
\begin{equation}
{{C}_{JBSPO-BH}}=I*K*N*O(utility),
\label{18}
\end{equation}
where $I$ is the number of iterations, and $O(utility)$ is the time complexity of the utility function determined by Eq.~\eqref{12}, $O(utility)=O({{K}^{2}})$. The value of $I$ is small due to JBSPO-BH algorithm's fast convergence as shown in Figure~\ref{fig3}.

In summary, the complexity ratio of the two algorithms can be obtained.
\begin{equation}
\begin{aligned}
  & \frac{{{C}_{JBSPO-BH}}}{{{C}_{GA-BH}}}\\
  &\approx \frac{I*K*N*O(fitness)}{G*P*O(utility)} \\ 
 &                  \approx \frac{I*K*N}{G*P}. \\ 
\end{aligned}
\label{19}
\end{equation}

The search efficiency of GA is low. In the simulation parameter setting of this work, the values of $G$ and $P$ are both large, and the performance of GA-BH algorithm is closed to that of the JBSPO-BH algorithm. Obviously, the complexity of the proposed algorithm is lower. As the scale of the problem rises, the search space of the GA-BH algorithm increases nonlinearly, while the search space of the proposed algorithm increases linearly. Therefore, the optimization of GA-BH algorithm will be more difficult and the computational complexity will be higher. That is, the JBSPO-BH algorithm in this paper is more likely to be applied to LEO satellites with limited onboard resources.

\section{conclusion}
\label{s5}

The BH controller deployed on the LEO satellite implements the JBSPO-BH algorithm to provide BH services for multiple sink nodes in different beam positions. The beam scheduling variables are binary integers, and the power allocation variables are continuous, so the problem is a mixed integer nonlinear problem. Then the problem is decoupled into two sub-problems: beam scheduling and power optimization. Firstly, it is proven that the beam scheduling problem is an exact potential game, and the NE point exists. Secondly, the optimization method is applied to solve the sub-problem of power allocation. The simulation results show that the JBSPO-BH algorithm can converge quickly. In terms of system performance, it is similar to the GA-BH algorithm but has lower time complexity and shorter execution time. Compared with the other comparison algorithms, the average throughput and fairness of the system are significantly improved.

In this work, we mainly consider the optimization of beam and power resources, but the carrier allocation and beamforming technology can also improve the performance in the BH LEO satellite systems. As there are differentiated services in LEO satellite systems, quality of service (QoS) assurance is considered in our future work. In addition, although the proposed algorithm executes once per BH cycle, it pays more attention to optimizing the system's short-term performance in each slot. Thus, the present research focus turns to long-term performance improvement in each BH cycle. We believe that the above directions are worth exploring.
\section*{ACKNOWLEDGEMENT}
\label{ACKNOWLEDGEMENT}

This work was supported by the National Key Research and Development Program of China 2021YFB2900504, 2020YFB1807900.

\appendix
\subsection{My Appendix A The Definition of potential game}
\label{app:a}
Combined with the objective function in the optimization problem Eq.~\eqref{9}, we define the potential function in this game as Eq.~\eqref{13}, where if $n\in {{({{a}_{k}},{{a}_{-k}})}}$, $\mathbf{I}(n\in {{({{a}_{k}},{{a}_{-k}})}})=1$, else $\mathbf{I}(n\in {{({{a}_{k}},{{a}_{-k}})}})=0$. For clarity, $\mathbf{I}(n\in ({{a}_{k}},{{a}_{-k}}))$ is written as $\mathbf{I}(n)$ , ${{P}_{ave}}H_{n,m}^{t}$ is written as $P_{n,m}^{t}$ and $BT_b$ is written as $B_b$. Besides, $\mathbf{L}(x)$ represents ${{\log }_{2}}(1+x)$. 
\begin{equation}
\begin{aligned}
  & F({{a}_{k}},{{a}_{-k}}) \\ 
 & ={{\sum\limits_{n=1}^{N}{\bigg({{B}_{b}}\mathbf{L}\Big(\frac{\mathbf{I}(n)P_{nn}^{t}}{\sum\limits_{m\ne n}{\mathbf{I}(m)P_{nm}^{t}+{{P}_{N}}}}\Big)-\hat{R}_{n}^{t}\bigg)}}^{2}}. \\
\end{aligned}
\label{13}
\end{equation}
\begin{definition}\label{definition1}In a game, if $\forall k\in \mathbf{K}$, there exists a function $F(a)$ satisfying the following relation, then the game is an exact potential game:
\begin{equation}
{{u}_{k}}(a_{k}^{'},{{a}_{-k}})-{{u}_{k}}({{a}_{k}},{{a}_{-k}})=F(a_{k}^{'},{{a}_{-k}})-F({{a}_{k}},{{a}_{-k}}).
\label{11}
\end{equation}
\end{definition}
In Eq.~\eqref{11}, $a_{k}^{'}$ represents the $k$th player's strategy that is different from ${{a}_{k}}$.

\subsection{My Appendix B The proof of theorem 1}

If the condition in Eq.~\eqref{11} is satisfied, $\mathbf{G}$ is an exact potential game.

Thus, we can get:
\begin{equation*}
\begin{aligned}
  & F({{a}_{k}},{{a}_{-k}}) \\ 
 & ={{\sum\limits_{n=1}^{N}{\bigg({{B}_{b}}\mathbf{L}\Big(\frac{\mathbf{I}(n)P_{nn}^{t}}{\sum\limits_{m\ne n}{\mathbf{I}(m)P_{nm}^{t}+{{P}_{N}}}}\Big)-\hat{R}_{n}^{t}\bigg)}}^{2}} \\ 
 & =\sum\limits_{n=1}^{N}{\Bigg(\bigg({{{{B}_{b}}\mathbf{L}\Big(\frac{\mathbf{I}(n)P_{nn}^{t}}{\sum\limits_{m\ne n}{\mathbf{I}(m)P_{nm}^{t}+{{P}_{N}}}}\Big)\bigg)}^{2}}} \\ 
 & -2{{B}_{b}\hat{R}_{n}^{t}}\mathbf{L}\Big(\frac{\mathbf{I}(n)P_{nn}^{t}}{\sum\limits_{m\ne n}{\mathbf{I}(m)P_{nm}^{t}+{{P}_{N}}}}\Big)+\hat{R}{{_{n}^{t}}^{2}}\Bigg) \\ 
 & =\sum\limits_{n=1}^{N}{\hat{R}{{_{n}^{t}}^{2}}}+\sum\limits_{k=1}^{K}{{{\bigg({{B}_{b}}\mathbf{L}\Big(\frac{P_{{{a}_{k}},{{a}_{k}}}^{t}}{\sum\limits_{l\ne k}{P_{{{a}_{k}},{{a}_{l}}}^{t}+{{P}_{N}}}}\Big)\bigg)}^{2}}}\\
 &-\sum\limits_{k=1}^{K}{2{{B}_{b}}\hat{R}_{{{a}_{k}}}^{t}\mathbf{L}\Big(\frac{P_{{{a}_{k}},{{a}_{k}}}^{t}}{\sum\limits_{l\ne k}{P_{{{a}_{k}},{{a}_{l}}}^{t}+{{P}_{N}}}}\Big)}\\
 & =C+u_k({{a}_{k}},{{a}_{-k}}). \\ 
\end{aligned}
\end{equation*}

From the above derivation, we can get
\begin{equation*}
\begin{aligned}
  & F({{a}_{k}},{{a}_{-k}})-F(a_{k}^{'},{{a}_{-k}}) \\ 
 & =C+u_k({{a}_{k}},{{a}_{-k}})-C-u_k(a_{k}^{'},{{a}_{-k}}) \\ 
 & =u_k({{a}_{k}},{{a}_{-k}})-u_k(a_{k}^{'},{{a}_{-k}}). \\ 
\end{aligned}
\end{equation*}

In summary, Theorem \ref{theo1} is proved.
\bibliography{myref}

\biographies

\begin{CCJNLbiography}{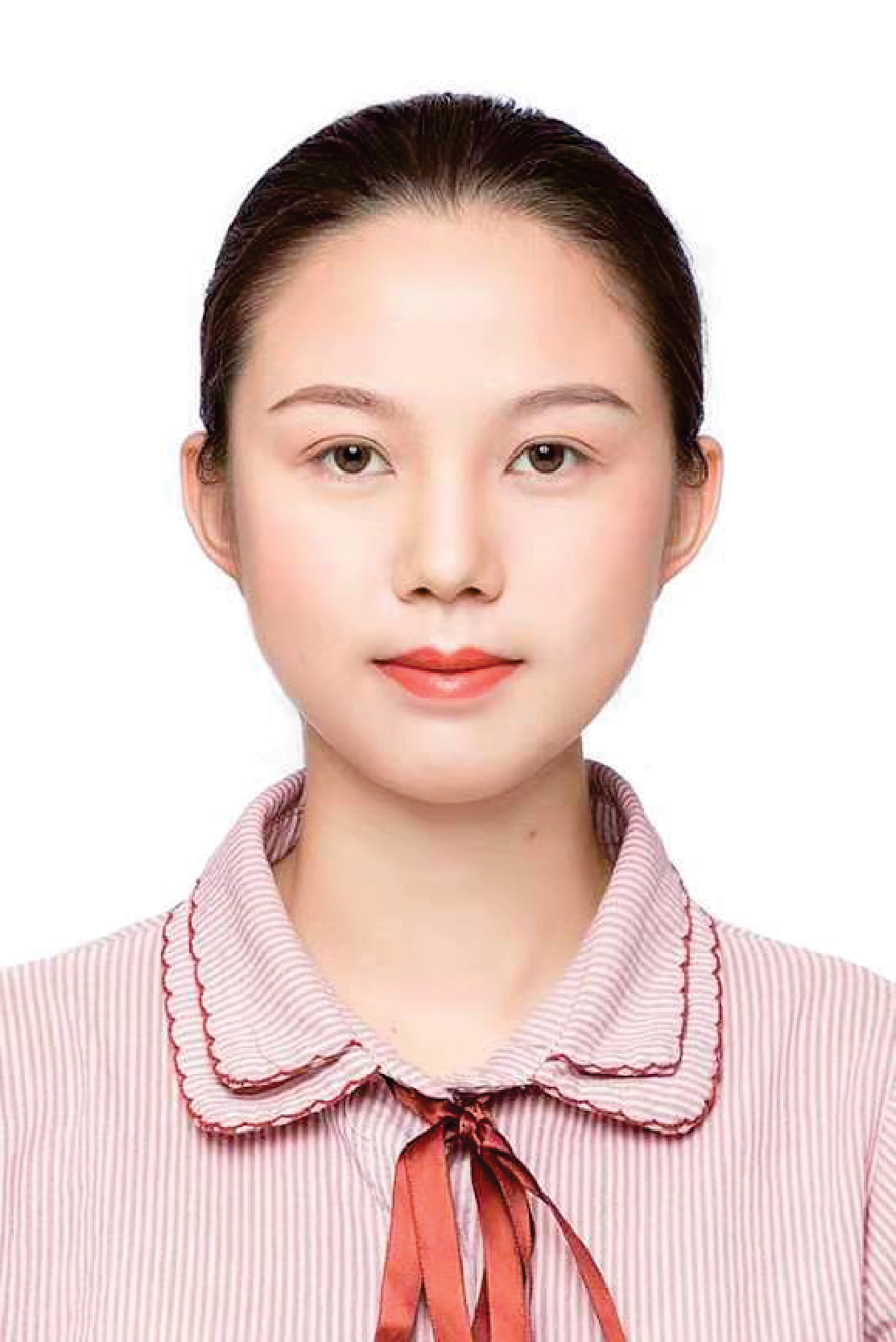}{Shuang Zheng}
received the B.E. degree in communication engineering from Beijing University of Posts and Telecommunications, Beijing, China, in 2021. She is currently pursuing the Ph.D. degree with the Key Laboratory of Universal Wireless Communications, School of Information and Communication Engineering, Beijing University of Posts and Telecommunications, Beijing, China. Her research interests include 5G/6G network technology and satellite terrestrial networks.
\end{CCJNLbiography}
\begin{CCJNLbiography}{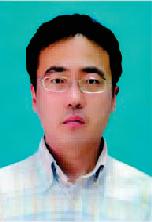}{Xing Zhang}
is Full Professor with the School of Information and Communications Engineering, Beijing University of Posts and
Telecommunications, China. His research interests are mainly in 5G/6G mobile communication system, mobile edge computing and data analysis, space-integrated-ground information network, cognitive radio and collaborative communication. He is a Senior Member of IEEE and IEEE ComSoc, Member of CCF.
\end{CCJNLbiography}
\begin{CCJNLbiography}{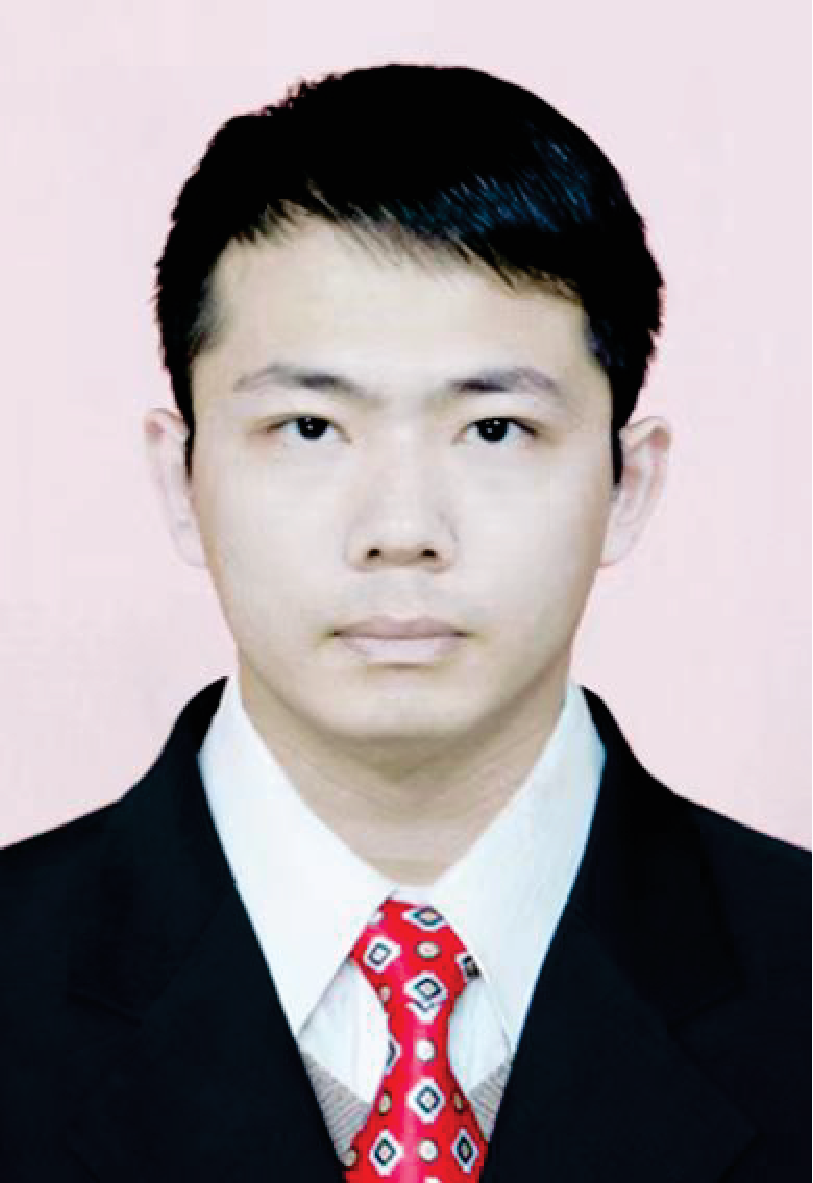}{Peng Wang}
received the B.E. degree in electronic and information engineering from Chang’an University, Xi’an, China, in 2012, and the M.E. degree in communication and information systems from Xi’an University of science and technology, in 2016. He is currently pursuing the Ph.D. degree with the Key Laboratory of Universal Wireless Communications, School of Information and Communication Engineering, Beijing University of Posts and Telecommunications, Beijing, China. His research interests include 5G network technology, satellite-terrestrial networks, and mobile edge computing.
\end{CCJNLbiography}
\begin{CCJNLbiography}{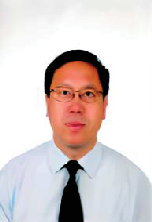}{Wenbo Wang}
received the B.S. degree in communication engineering and the M.S. and Ph.D. degrees in signal and information processing from the Beijing University of Posts and Telecommunications (BUPT), Beijing, China, in 1986, 1989, and 1992, respectively. From 1992 to 1993, he was a Researcher with ICON Communication Inc., Dallas, TX, USA. His research interests include transmission technology, broadband wireless access, wireless network theory, digital signal processing, multiple-input-multiple-output (MIMO), cooperative and cognitive communications, and software radio technology.
\end{CCJNLbiography}

\end{document}